\documentclass[12pt, oneside]{amsart}

\usepackage[margin=1.25in]{geometry}
\usepackage[round]{natbib}
\usepackage{setspace}
\usepackage{amsmath}
\usepackage{amsfonts}
\usepackage{amssymb}
\usepackage{amsthm}
\usepackage{hyperref}
\usepackage{cleveref}
\usepackage{enumerate}
\usepackage{graphicx}
\usepackage{relsize}

\newtheorem*{theorem*}{Theorem}

\newtheorem{theorem}{Theorem}
\newtheorem*{lemma*}{Lemma}
\newtheorem{lemma}{Lemma}
\newtheorem*{proposition*}{Proposition}
\newtheorem{proposition}{Proposition}
\newtheorem*{corollary*}{Corollary}
\newtheorem{corollary}{Corollary}

\newtheorem*{fact*}{Fact}

\theoremstyle{definition}

\newtheorem{definition}{Definition}
\newtheorem{example}{Example}

\crefname{lemma}{Lemma}{Lemmas}

\title{Connected Incomplete Preferences}

\author{Leandro Gorno \phantom{.....} Alessandro T. Rivello\\ \smaller \phantom{.} FGV EPGE \phantom{.............................} FGV EPGE \phantom{.........}}

\date{First draft: January 2020. This version: April 2026.}

\begin{document}

\begin{abstract}
This paper explores a new class of incomplete preferences---termed ``connected preferences''---in which maximal domains of comparability are topologically connected. We provide necessary and sufficient conditions for continuous preferences to be connected. We also characterize their maximal domains of comparability. Our results extend classical findings in decision theory by linking topological properties of the choice space with the structure of preferences, offering a novel perspective on incompleteness in economic models.
\vspace{10pt}\\
\noindent \textit{Keywords:} incomplete preferences, maximal elements, continuity, connectedness.\vspace{10pt}\\
\noindent \textit{JEL classifications:} C61, D81.
\end{abstract}

\maketitle

\onehalfspacing

\vspace{10pt}

\section{Introduction}
The standard model of choice in economics is the maximization of a complete and transitive preference relation over a fixed set of alternatives. While completeness of preferences is often considered a strong assumption, it is challenging to weaken it and preserve enough structure for the resulting model to yield interesting results. This paper contributes to this effort by studying ``connected preferences''---preferences that may be incomplete but possess connected ``maximal domains'' of comparability.

The concept of a maximal domain of a preference---a subset of the choice space on which the preference is complete, and which cannot be extended while preserving completeness---is studied extensively in \citet*{Gorno2018}. That work establishes that such domains cover the space of alternatives and, crucially, that the maximal domains of continuous preferences are necessarily closed. This paper builds on those results to investigate under which conditions continuous preferences are connected.

Connected preferences are particularly interesting because the nature of their incompleteness is intrinsically tied to the topology of the choice space. Such a link means that the indecision associated with the lack of completeness is not arbitrary: it reflects---and is constrained by---the alternatives themselves. This opens the door to a richer interplay between economic reasoning and topological intuition, where notions like continuity and connectedness acquire direct behavioral meaning. By grounding incompleteness in the structure of the choice space, we gain a framework capable of capturing more realistic decision environments while still supporting powerful analytical tools.

We offer four new results. Theorem~\ref{th:necessity} identifies a basic necessary condition for a continuous preference to be connected in the sense above, while Theorem~\ref{th:sufficiency} provides sufficient conditions when the space of alternatives is compact. Building on the latter, Theorem~\ref{th:characterization} characterizes the maximal domains of comparability. Finally, Theorem~\ref{th:sufficiency_arc} proves that, in compact metric spaces, continuous antisymmetric preferences without jumps have arc-connected maximal domains.

Methodologically, we adopt an incomplete preference perspective on the theoretical literature relating assumptions about preferences with the structure of the space of alternatives. For example, \citet*{Schmeidler1971} shows that every nontrivial preference on a connected topological space which satisfies seemingly innocuous continuity conditions must be complete. \citet*{KhanUyanik2021} revisit Schmeidler's theorem and connect it to earlier results by \citet*{Eilenberg1941}, \citet*{Sonnenschein1965}, and \citet*{Sen1969}. They provide a thorough analysis of the logical relationships among continuity, completeness, transitivity, and the connectedness of the space.

In particular, Theorem 1 in \citet*{KhanUyanik2021} implies a converse to Schmeidler's theorem: if every nontrivial ``Schmeidler preference'' is complete, the underlying space must be connected. We offer a different kind of converse: any compact space that admits a complete Schmeidler preference with connected indifference classes and no ``jumps'' must be connected.

The rest of the paper is organized as follows. Section~\ref{sec:connected} introduces the concept of connected preferences and establishes necessary and sufficient conditions for continuous preferences to be connected. Section~\ref{sec:domains} uses the concept of connected preferences to characterize the maximal domains of comparability. Section~\ref{sec:arc_connected} investigates the stronger property of arc-connectedness, examining compact metric spaces. Finally, Section~\ref{sec:applications} presents some applications of the results. All proofs are relegated to Appendix~\ref{app:proofs}.

\section{Connected preferences}
\label{sec:connected}

\subsection{Preliminaries}
Let $X$ be a (nonempty) set of alternatives equipped with some topology. A \textit{preference} is a reflexive and transitive binary relation on $X$. For the rest of this paper, we consider a fixed preference $\succsim$, with $\sim$ and $\succ$ denoting the symmetric and asymmetric parts of $\succsim$, respectively. A preference $\succsim$ is \textit{trivial} if there are no $x,y \in X$ such that $x \succ y$. If $\succsim$ is not trivial, then it is \textit{nontrivial}.

We recall some standard topological definitions. A set $A \subseteq X$ is \textit{connected} if it cannot be partitioned into two disjoint nonempty sets that are closed relative to $A$. A set $A \subseteq X$ is \textit{arc-connected} if, for every pair of distinct points $x,y \in A$, there exists a continuous injection $f: [0,1] \to A$ such that $f(0)=x$ and $f(1)=y$; the image $f([0,1])$ is called an \textit{arc} from $x$ to $y$. Finally, $X$ is \textit{uniquely arc-connected} if it is Hausdorff and for any distinct $x,y \in X$, there is a unique arc connecting them.

$\succsim$ is \textit{continuous} if $\left\{y \in X\middle| y \succsim x \right\}$ and $\left\{y \in X\middle| x \succsim y \right\}$ are closed sets for every $x \in X$. $\succsim$ has \textit{connected indifference classes} if $\left\{y \in X\middle| y \sim x \right\}$ is connected for every $x \in X$. A set $A \subseteq X$ \textit{contains every indifferent alternative} if $x \in A$, $y \in X$, and $x \sim y$ implies $y \in A$. A set $A \subseteq X$ \textit{has no exterior bound} if whenever $x \in X$ satisfies either $x \succsim y$ for all $y \in A$ or $y \succsim x$ for all $y \in A$, then $x \in A$. 

$\succsim$ is \textit{complete} on a set $A \subseteq X$ if, for every $x,y \in A$, either $x \succsim y$ or $y \succsim x$ holds. The set $A$ is a \textit{domain} if $\succsim$ is complete on $A$. If $A$ is a domain such that there exists no larger domain containing it, then $A$ is a \textit{maximal domain}. A maximal domain is \textit{trivial} if it is an indifference class and \textit{nontrivial} if it is not.

\subsection{Main definition}

Maximal domains are studied in \citet*{Gorno2018}, where he shows that if $\succsim$ is continuous, the maximal domains are closed in $X$.\footnote{See Theorem 1 in \citet*{Gorno2018}.} The main concept of this paper targets another topological property of the maximal domains:

\begin{definition}
Preference $\succsim$ is \textit{connected} if every maximal domain is connected.
\end{definition}

This definition provides an intuitive way of linking decision makers' indecisiveness with the topological disconnection of the space of alternatives. For if $\succsim$ is a connected preference and two alternatives $x$ and $y$ are comparable according to $\succsim$, then both $x$ and $y$ must belong to a connected subset of $X$ with the property that every two alternatives in the set are comparable according to $\succsim$. 

We will restrict attention to preferences that are not only connected, but also continuous. As a result, in this paper, maximal domains are necessarily closed subsets of $X$ and the question of interest is whether they are also connected.

\subsection{A necessary condition}

A natural first step towards a characterization of connected preferences is to obtain a simple necessary condition.

\begin{definition}
A \textit{jump} of preference $\succsim$ is a pair of alternatives $(x,y) \in X \times X$ such that $x \succ y$ and there is no $z \in X$ satisfying $x \succ z \succ y$.
\end{definition}

The notion of preferences without jumps is not new; its content coincides with a well-known definition of order-denseness for sets.\footnote{The set $X$ is said to be $\succsim$\textit{-dense} if for every $x,y \in X$ satisfying $x \succ y$ there exists $z \in X$ such that $x \succ z \succ y$ (see \citet*{Ok2007}, p. 92). Evidently, $X$ is $\succsim$-dense if and only if $\succsim$ has no jumps. We should note that alternative notions of order-denseness exist in the literature.} Our first result shows that connected continuous preferences cannot have jumps.

\begin{theorem}
\label{th:necessity}
If $\succsim$ is continuous and connected, then $\succsim$ has no jumps.
\end{theorem}

In particular, Theorem~\ref{th:necessity} implies that, when the space of alternatives is connected, preferences admitting a continuous utility representation cannot have jumps. However, it is easy to see that not every continuous preference without jumps is connected:

\begin{example}
Consider $X=[-1,1]$ equipped with the natural topology and let $\succsim \hspace{2pt} = \left\{(x,y)\in X^2 \middle| x = y \vee x^2 = y^2 =1\right\}$. Then, the preference $\succsim$ is continuous and has no jumps, but is not connected (the maximal domain $\{-1, 1\}$ is not a connected set).
\end{example}

\subsection{A sufficiency theorem}

We already know that continuous and connected preferences cannot have jumps. In this subsection, we provide a set of assumptions which constitute a sufficient condition for a preference to be connected.

\begin{theorem}
\label{th:sufficiency}
If $X$ is compact and $\succsim$ is a continuous preference that has connected indifference classes and has no jumps, then $\succsim$ is connected.
\end{theorem}

While the full proof is in the Appendix, here we provide some intuition for why this is true. If a maximal domain were disconnected, it could be partitioned into two compact sets, say $A$ and $B$. Since the preference is continuous and complete on the domain, it would be possible to get a best element in one set and a worst element in the other that are ``adjacent" in the preference order. If these boundary elements were indifferent, it would be possible to split their connected indifference class between $A$ and $B$, yielding a contradiction. If one is strictly preferred to the other, the lack of jumps ensures the existence of an intermediate alternative $z$. We show in the proof that adding $z$ to the domain maintains completeness, contradicting the maximality of the original domain.

The following example identifies an important class of connected preferences:

\begin{example}
Let $X$ be the set of Borel probability measures (lotteries) on a compact metric space of prizes $Z$, equipped with the topology of weak convergence. Following \citet*{Dubra2004}, we say that the preference $\succsim$ is an \textit{expected multi-utility preference} if there exists a set $\mathcal{U} \subseteq \mathbb{R}^Z$ of continuous functions such that $x \succsim y$ if and only if
\[
\int_Z u dx \geq \int_Z u dy
\]
holds for all $u \in \mathcal{U}$. It is easy to verify that all the assumptions of Theorem~\ref{th:sufficiency} hold. Thus, $\succsim$ is connected.
\end{example}

It should be stressed that the linearity of expected multi-utility preferences in the previous example is key to its connectedness. One way of appreciating this fact is to consider preferences that admit a finite, continuous, and strictly quasiconcave multi-utility representation. Such preferences are interesting because they yield optimal choices that vary continuously with the budget set.\footnote{Theorem~1 in \citet{GornoRivello2023} provides a maximum theorem for such preferences for the case of compact $X$.} However, a preference satisfying these conditions may fail to be connected, as the following example demonstrates.

\begin{example}
\label{ex:quasiconcave}
    Let $X=[-1,1]$ and let $\succsim$ be the preference represented by $\{u,v\} \subseteq \mathbb{R}^X$ with $u(x)=1-x^2$ and $v(x)=x-x^2$. Both $u$ and $v$ are continuous and strictly quasiconcave. The indifference classes of $\succsim$ are singletons, thus connected. Moreover, $u(1)=u(-1)=0$ and $v(1)=0>-2=v(-1)$ imply that $1 \succ -1$. However, there is no $x \in X$ such that $1 \succ x \succ -1$ (for $1 \succ x$ implies $u(1) \geq u(x)$, which is equivalent to $x^2 \geq 1$ and so can only be satisfied at $x \in \{1,-1\}$). This means that $(1,-1)$ is a jump of $\succsim$. Thus, by Theorem~\ref{th:necessity}, $\succsim$ is not connected.
\end{example}

\begin{figure}[h]
    \centering
    \includegraphics[scale=0.5]{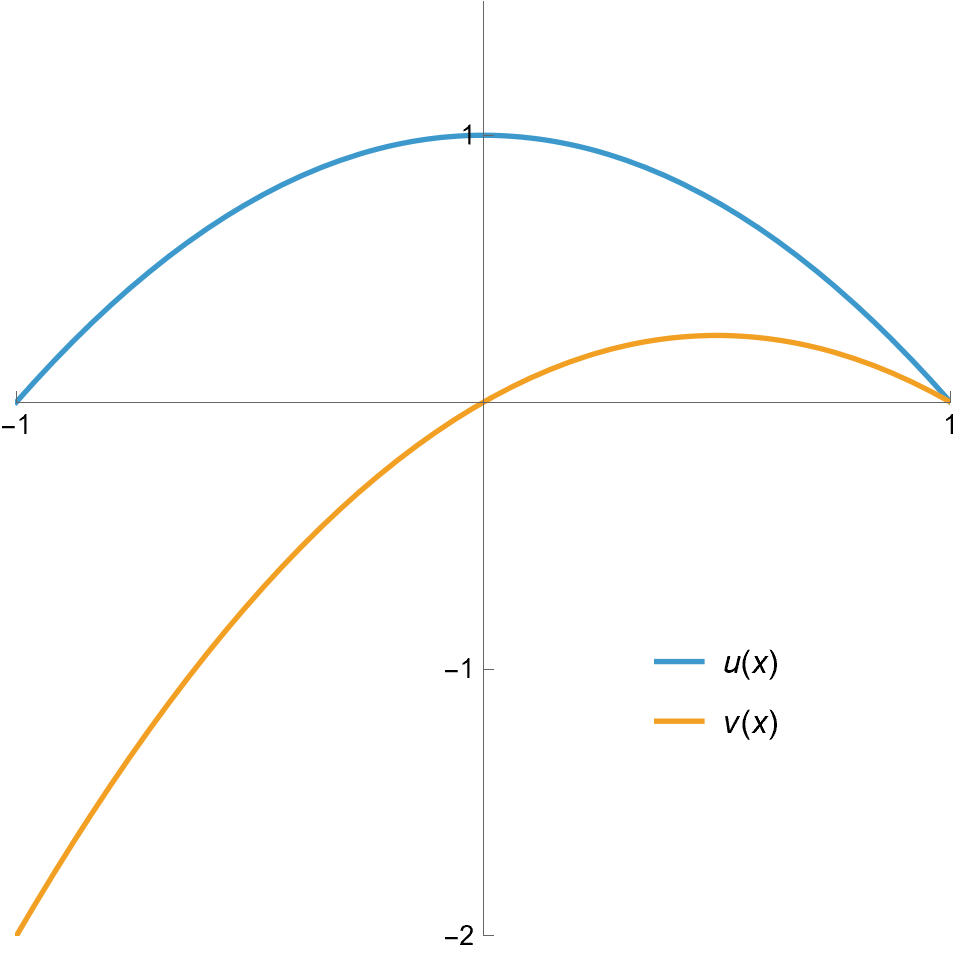}
    \caption{Finite multi-utility preference in Example~\ref{ex:quasiconcave}}
    \label{fig:quasiconcave}
\end{figure}

Example~\ref{ex:quasiconcave} shows that even quite strong restrictions on general multi-utility outside of the lottery framework may fail to ensure connected maximal domains. While expected multi-utility remains an important example in which the assumptions of Theorem~\ref{th:sufficiency} hold, it is worth noting that these assumptions are satisfied in more general models. For instance, assuming that $X$ is convex, we say that $\succsim$ satisfies \textit{betweenness} if $x \succsim y$ implies $x \succsim \alpha x+(1-\alpha)y \succsim y$ for all $x,y \in X$ and $\alpha \in [0,1]$.\footnote{This axiom was introduced by \citet*{Dekel1986}. Prominent examples of preferences satisfying betweenness include preferences satisfying the independence axiom (such as expected utility or the expected multi-utility preferences studied in \citet*{Dubra2004}) and also preferences exhibiting disappointment aversion as in \citet*{Gul1991}.} The following lemma shows that betweenness is sufficient for two of the three preference conditions in Theorem~\ref{th:sufficiency}.

\begin{lemma}
\label{lemma:betweenness}
If $X$ is convex and $\succsim$ is a continuous preference that satisfies betweenness, then $\succsim$ has connected indifference classes and no jumps.
\end{lemma}

As a consequence, the application of  Theorem~\ref{th:sufficiency} to preferences satisfying betweenness is quite direct. Still, it is worth noting that the reach of Theorem~\ref{th:sufficiency} goes beyond, allowing us to establish the connectedness of preferences that violate betweenness.

\begin{example}
\label{ex:fibers}
Let $X = [0,1]^2$. For each $\gamma \in (0,\infty)$, let 
\[
C_\gamma := \left\{(t, t^\gamma) \in X \mid t \in [0,1]\right\}
\]
be the graph of a power function from the unit interval onto itself with exponent $\gamma$. We also include the limiting cases $C_0 = \left(\{0\} \times [0,1]\right) \cup \left([0,1] \times \{1\}\right)$ and $C_\infty = \left([0,1] \times \{0\}\right) \cup \left(\{1\} \times [0,1]\right)$, describing the left-upper and bottom-right boundaries of the unit square, respectively. Define a preference $\succsim$ on $X$ as:
\[
\succsim \hspace{2pt} := \left\{(x,y)\in X \times X\middle| x_1 \geq y_1 \wedge x_2 \geq y_2 \wedge \exists \gamma \in [0, \infty]: \{x,y\}\subseteq C_\gamma \right\}.
\]
$\succsim$ is continuous, antisymmetric (so it has connected indifference classes), and has no jumps. Thus, by Theorem~\ref{th:sufficiency}, $\succsim$ is connected. 
Crucially, this preference violates betweenness. Consider any $\gamma \neq 1$ (so $C_\gamma$ is the graph of the strictly convex or strictly concave function $t \to t^\gamma$). For any distinct $x, y \in C_\gamma$, the convex combination $z = \alpha x + (1-\alpha)y$ (for $\alpha \in (0,1)$) does not lie on $C_\gamma$. Since the curves $\{C_\gamma\}_\gamma$ partition $X \setminus \{(0,0), (1,1)\}$, $z$ belongs to a different curve $C_{\gamma'}$. Thus, $z$ is incomparable to $x$ and $y$, violating the condition $x \succsim z \succsim y$.
\end{example}

Finally, we should note that Theorem~\ref{th:sufficiency} relies on the restrictive assumption that $X$ is compact. Compactness of the ambient space is not necessary for a preference to be connected. For example, if $X$ is connected, every complete preference on $X$ is connected, irrespective of whether $X$ is compact or not. However, the compactness of $X$ cannot be dropped or replaced by connectedness in Theorem~\ref{th:sufficiency}. Indeed, as the following examples suggest, in every non-compact connected subset of $\mathbb{R}$, there are continuous antisymmetric preferences without jumps that are not connected.

\begin{example}
\label{ex:interval_gap}
    Let $X=[0,1)$ and define
    \[
    \succsim \hspace{2pt} := \left\{(x,y) \in X^2 \middle| x=y  \vee \left(y=0 \wedge x \geq 1/2\right) \vee 1/2 \leq x \leq y \right\}.
    \]
    The preference $\succsim$ is obviously antisymmetric. Moreover, $\succsim$ is continuous as the sections $\left\{y \in X \middle| y \succsim x\right\}$ and $\left\{y \in X \middle| x \succsim y\right\}$ are either singletons or sets of the forms $[1/2,x]$, $[1/2,1) \cup \{0\}$, or $[x,1) \cup \{0\}$, which are closed in $X$. Finally, $\succsim$ has no jumps since for all $x,y \in X$ that could satisfy $x \succ y$, we can set $z = (x+y)/2$ when $y>x\geq 1/2$ and $z = (x+1)/2$ when $y=0$. In each case, we have $x \succ z \succ y$.
\end{example}

\begin{example}
\label{ex:interval_gap_2}
    Let $X=(0,1)$ and define
    \[
    \succsim \hspace{2pt} := \left\{(x,y) \in X^2 \middle| x=y  \vee \left(y \leq 1/3 \wedge x \geq 2/3\right) \vee x \leq y \leq 1/3 \vee 2/3 \leq x \leq y\right\}.
    \]
    The preference $\succsim$ is obviously antisymmetric. Moreover, $\succsim$ is continuous as the sections $\left\{y \in X \middle| y \succsim x\right\}$ and $\left\{y \in X \middle| x \succsim y\right\}$ are either singletons or sets of the forms $(0,x] \cup [2/3,1)$, $[2/3,x]$, $[x,1/3]$, or $(0,1/3] \cup [x,1)$, which are closed in $X$. Finally, $\succsim$ has no jumps since for all $x,y \in X$ that could satisfy $x \succ y$, we can set $z = (x+y)/2$ when $x,y \in (0,1/3]$ or $x,y \in [2/3,1)$, and we can set $z=(x+1)/2$ when $x \in [2/3,1)$ and $y \in (0,1/3]$. In each case, we have $x \succ z \succ y$.
\end{example}

Examples~\ref{ex:interval_gap} and \ref{ex:interval_gap_2} demonstrate that, without compactness, connectedness of the preference is not guaranteed by the other assumptions in Theorem~\ref{th:sufficiency}, even if we limit ourselves to connected ambient spaces. In fact, these cases point to a more general topological phenomenon: without compactness, one can often construct continuous antisymmetric preferences with no jumps that fail to be connected. As shown formally in Appendix~\ref{app:compactness}, compactness is indeed a necessary condition for the connectedness of all such preferences in most topological spaces of interest in Economics.

\section{Characterization of maximal domains}
\label{sec:domains}

Building on Theorem~\ref{th:sufficiency}, we can offer a useful characterization of maximal domains:

\begin{theorem}
\label{th:characterization}
Assume $X$ is compact and $\succsim$ is continuous, has connected indifference classes, and has no jumps. Then, a set $A \subseteq X$ is a maximal domain if and only if it is a connected domain that contains every indifferent alternative and has no exterior bound.
\end{theorem}

The intuition behind the necessity in this characterization is straightforward. Given the assumptions, Theorem~\ref{th:sufficiency} ensures that any maximal domain must be connected. The other two conditions are also required to ensure maximality: if a domain did not contain every indifferent alternative, it could be extended by adding any of the missing indifferent points; if it had an exterior bound, it could be extended by adding the bound. The sufficiency asserted by Theorem~\ref{th:characterization} also rests on a simple idea: the three conditions together leave no room for a domain to be extended. Containing every indifferent alternative and having no exterior bound already ``seal off'' the domain from the outside, and connectedness ensures there are no internal gaps either, implying the domain is already as large as it can be.

We finish this section discussing the two additional assumptions employed in Theorem~\ref{th:characterization}: the compactness of $X$ and that $\succsim$ has connected indifference classes.

\subsection{Compactness}
Compactness of $X$ cannot be dispensed with, as the following example shows.

\begin{example}
Let $X=\{-1\}\cup [0,1)$ and $\succsim \hspace{2pt} = \left\{(x,y) \in X^2 \middle| x=-1 \vee x \geq y \geq 0 \right\}$. Then, $X$ is bounded, locally compact and $\sigma$-compact, but fails to be compact. Moreover, $\succsim$ is complete, continuous, and has no jumps. However, the only maximal domain is $X$ itself and is not connected.
\end{example}

Example~\ref{ex:interval_gap} in Section~\ref{sec:connected} shows that the characterization does not generally hold without compactness even if we confine ourselves to connected ambient spaces. Proposition~\ref{prop:non-compact} in Appendix~\ref{app:compactness} formally establishes that compactness is necessary to obtain the characterization in Theorem~\ref{th:characterization} for a broad class of topological spaces.

\subsection{Connected indifference classes}

On the one hand, the assumption that indifference classes are connected is not strictly necessary for the conclusion of Theorem~\ref{th:characterization}. That is, there are examples that violate this condition where the equivalence in the theorem holds:

\begin{example}
\label{ex:disconnected_indifference1} 
Let $X=[-1,1]$ and $\succsim \hspace{2pt}=\left\{(x,y) \in X^2 \middle| x^2 \geq y^2\right\}$. Since $\succsim$ is complete and $X$ is connected, $\succsim$ is connected, even though all indifference classes but $\{0\}$ are disconnected.
\end{example}

On the other hand, it is a tight condition: there are examples that violate it, satisfy the remaining conditions, and for which the equivalence in the theorem fails to hold:

\begin{example}
\label{ex:disconnected_indifference2}
Let $X=\{-1\}\cup[0,1]$ and $\succsim \hspace{2pt}=\left\{(x,y) \in X^2 \middle| x^2 \geq y^2\right\}$.
\end{example}

Examples~\ref{ex:disconnected_indifference1} and \ref{ex:disconnected_indifference2} highlight that the characterization in Theorem~\ref{th:characterization} could potentially be improved by weakening the assumptions of connected indifference classes. However, we leave the task of refining the characterization in this direction for future research and move to investigate when the conclusion of connectedness that Theorem~\ref{th:sufficiency} obtains under that assumption can be strengthened to arc-connectedness, a property that ensures that, when two alternatives are comparable, there is a transition between them over which preferences are continuously expressed.

\section{Arc-connected preferences}
\label{sec:arc_connected}

In many economic applications, it is useful to know not only that a domain is connected (``in one piece''), but that one can move continuously from one alternative to another while remaining within the domain of comparability. This property, known as arc-connectedness, is strictly stronger than standard connectedness.

\begin{definition}
$\succsim$ is \textit{arc-connected} if every maximal domain is arc-connected.
\end{definition}

In the particular case of antisymmetric preferences (\textit{i.e.}, partial orders) on a metrizable space, we can strengthen the conclusion of Theorem~\ref{th:sufficiency}:

\begin{theorem}
\label{th:sufficiency_arc}
If $X$ is a compact second-countable space and $\succsim$ is a continuous antisymmetric preference with no jumps, then $\succsim$ is arc-connected.
\end{theorem}

The argument for Theorem~\ref{th:sufficiency_arc} rests on the existence of a ``local" continuous utility representations in every maximal domain. Due to antisymmetry, this representation can be inverted, allowing us to construct an arc between any two comparable points.

Like Theorem~\ref{th:sufficiency}, Theorem~\ref{th:sufficiency_arc} assumes that $X$ is compact. As Example~\ref{ex:interval_gap} in Section~\ref{sec:connected} shows, even in a uniquely arc-connected space such as $[0,1)$, lack of compactness allows for maximal domains that fail to be connected (and, thus, fail to be arc-connected as well). Moreover, Theorem~\ref{th:sufficiency_arc} further assumes that $X$ is second-countable and that $\succsim$ is antisymmetric. The second-countability of $X$ is a relatively mild technical condition\footnote{In the context of Theorem~\ref{th:sufficiency_arc}, if $X$ is metrizable then it is also second-countable.} used to ensure that, in every domain of $\succsim$, it is possible to define a continuous utility representation of the restriction of $\succsim$ to that domain. Finally, the following example shows that the antisymmetry assumption in Theorem~\ref{th:sufficiency_arc} cannot be relaxed to connected indifference classes.

\begin{example}
    \label{ex:topologist_curve}
    Let $X=[0,1]\times [-1,1]$ and define the ``closed topologist's sine curve" $C := \left\{(x_1,x_2) \in X \middle| x_1 \in (0,1] \wedge x_2 = \sin(1/x_1)\right\} \cup \{0\}\times [-1,1]$. $C$ is closed and connected but not arc-connected. Consider the preference given by
    \[
    \succsim \hspace{2pt} := \left\{(x,y)\in X^2\middle| x=y \vee x,y \in C\right\}.
    \]
    Clearly, $X$ is compact and $\succsim$ is continuous and has no jumps (there are no points $x,y \in X$ such that $x \succ y$). Moreover, $\succsim$ has connected indifference classes ($C$ is the only nontrivial indifference class). Hence, by Theorem~\ref{th:sufficiency}, all maximal domains are connected. Note also that $X$ is compact and second-countable, but $\succsim$ is not antisymmetric. This means that all conditions for Theorem~\ref{th:sufficiency_arc} except antisymmetry hold. However, the  maximal domain $C$ is not arc-connected, so the conclusion of Theorem~\ref{th:sufficiency_arc} fails.
\end{example}

\section{Applications}
\label{sec:applications}

\subsection{A maximum theorem for incomplete preferences}

Theorem~\ref{th:characterization} is directly used in the proof of Theorem 4 in \citet*{GornoRivello2023}, a result that identifies conditions on primitives under which ``domain continuity" (the requirement that limits of maximal domains are themselves maximal domains) is equivalent to the validity of a maximum theorem (so that limit points of maximal and minimal elements are themselves maximal and minimal, respectively). In particular, the characterization of maximal domains provided by Theorem~\ref{th:characterization} is critical in establishing that ``domain continuity" is necessary for a maximum theorem to hold in this setting.

\subsection{First-order stochastic dominance}

Suppose $X$ is the set of cumulative distribution functions (CDFs) over a compact interval $[0,\overline{z}]$ (endowed with the topology of weak convergence of the associated probability measures). Let $\geq_1$ denote the first-order stochastic dominance relation on $X$, that is, $F \geq_1 G$ if and only if $F(z) \leq G(z)$ for all $z \in [0,\overline{z}]$.

\begin{proposition}
\label{prop:arc-connected}
$\geq_1$ is arc-connected. Moreover, a subset of $X$ is a maximal domain of $\geq_1$ if and only if it is the image of a $\geq_1$-increasing arc joining the degenerate CDFs associated with $0$ and $\overline{z}$.
\end{proposition}

An analogous result holds for second-order stochastic dominance.

\subsection{Schmeidler preferences}

\citet*{Schmeidler1971} shows that, in a connected space, every nontrivial preference satisfying seemingly innocuous continuity conditions must be complete. In this section, we explore the implications of these assumptions in spaces that are not connected. 

We start by defining the class of preferences which are the subject of Schmeidler's theorem:

\begin{definition}
A preference $\succsim$ is a \textit{Schmeidler preference} if it is continuous and the sets $\left\{y \in X \middle| x \succ y\right\}$ and $\left\{y \in X \middle| y \succ x\right\}$ are open for all $x\in X$.
\end{definition}

The following definition captures a property that generalizes the conclusion of Schmeidler's theorem in terms of maximal domains:

\begin{definition}
A preference is \textit{decomposable} if every maximal domain is either a connected component or an indifference class.
\end{definition}

Note that, when $X$ is connected, every nontrivial decomposable preference is complete. More generally, any two distinct maximal domains of a decomposable preference must necessarily be disjoint. As a result, if a decomposable preference is locally nonsatiated, then no maximal domain can be trivial or, equivalently, every maximal domain must be a connected component. 

We can now state the main result of this section:

\begin{proposition}
\label{prop:Schmeidler}
Let $X$ be compact and let $\succsim$ be a Schmeidler preference with connected indifference classes. Then, $\succsim$ is decomposable if and only if $\succsim$ has no jumps.
\end{proposition}

Note that every continuous and complete preference is a Schmeidler preference. In that particular case, we have the following:

\begin{corollary}
\label{corollary:complete}
Let $X$ be compact and let $\succsim$ be a continuous and complete preference with connected indifference classes. Then, $X$ is connected if and only if $\succsim$ has no jumps.
\end{corollary}

\citet*{Schmeidler1971} shows that if $X$ is connected, then every nontrivial Schmeidler preference must be complete. \citet*{KhanUyanik2021} prove a converse and obtain the following characterization: $X$ is connected if and only if every nontrivial Schmeidler preference is complete. Corollary~\ref{corollary:complete} above implies a different characterization of topological connectedness for compact spaces:
if $X$ is compact, then $X$ is connected if and only if there exists at least one complete and continuous preference on $X$ with connected indifference classes and no jumps.\footnote{One direction of the equivalence is immediate: if $X$ is connected, the trivial preference that declares all alternatives indifferent satisfies all the desired properties.}

\section*{Acknowledgments}
This paper is based on material contained in Chapter 3 of Alessandro Rivello's PhD dissertation at FGV EPGE. We thank Guilherme Araújo Lima for excellent research assistance. We are also grateful to the editor and two anonymous referees, as well as Jos\'{e} Heleno Faro, Kazuhiro Hara, Paulo Klinger Monteiro, Lucas Maestri, Gil Riella, and participants at various seminars for their helpful comments. This study was financed in part by the Coordenação de Aperfeiçoamento de Pessoal de N\'{i}vel Superior - Brasil (CAPES) - Finance Code 001. Leandro Gorno also gratefully acknowledges the financial assistance of the Brazilian National Council for Scientific and Technological Development (CNPq).

\appendix

\section{Proofs}
\label{app:proofs}

\begin{proof}[Proof of Theorem~\ref{th:necessity}]
Suppose, seeking a contradiction, that $\succsim$ has a jump: there exist alternatives $x,y \in X$ such that $x \succ y$ and no $z \in X$ satisfies $x \succ z \succ y$. By Lemma 1 in \citet{Gorno2018}, there exists a maximal domain $D$ such that $\{x,y\} \subseteq D$. Define $A := \left\{z \in D\middle| z \succsim x\right\}$ and $B := \left\{z \in D\middle| y \succsim z\right\}$. Clearly, $A$ and $B$ are nonempty, $A \cap B= \emptyset$, and $A \cup B = D$. Moreover, since $\succsim$ is continuous, $A$ and $B$ are closed relative to $D$. It follows that $D$ is not connected, a contradiction.
\end{proof}

\begin{proof}[Proof of Theorem~\ref{th:sufficiency}]
Suppose, seeking a contradiction, that there is a maximal domain $D$ that is not connected. Then, there exist disjoint nonempty sets $A$ and $B$ such that $A \cup B = D$ and both are closed relative to $D$. Since $\succsim$ is continuous, Proposition 1 in \citet{Gorno2018} implies that $D$ is closed (in $X$), hence $A$ and $B$ are also closed. Moreover, since $X$ is compact, $A$ and $B$ are compact as well. 

Let $\overline{x}_A$ and $\overline{x}_B$ be the best elements in $A$ and $B$, respectively.\footnote{It is well known that a preference with closed upper sections has a best element on every compact subset of $X$. To the best of our knowledge, the first explicit statement of this fact is Theorem 5.3.4 in \cite{Rader1972}.} Since $D$ is a domain, $\overline{x}_A$ and $\overline{x}_B$ are comparable. Suppose first that $\overline{x}_A \sim \overline{x}_B$ and consider the indifference class $I := \left\{ x \in X \middle| x \sim \overline{x}_A \right\}$. Note that $I \subseteq D$, because $D$ is a maximal domain. Hence, the sets $I \cap A$ and $I \cap B$ form a partition of $I$ into two nonempty closed sets, which contradicts the assumption that $\succsim$ has connected indifference classes. 

Suppose now that $\overline{x}_A \succ \overline{x}_B$ (the remaining case is symmetric). Define the set $C:=\left\{x \in A\middle| x \succsim \overline{x}_B\right\}$. $C$ is nonempty (as $\overline{x}_A \in C$) and compact. Let $\underbar{x}_C$ be the worst element in $C$. Since $\underbar{x}_C \in C$, we have $\underbar{x}_C \succsim \overline{x}_B$. We claim that $\underbar{x}_C \succ \overline{x}_B$. If $\underbar{x}_C \sim \overline{x}_B$, then $\underbar{x}_C$ and $\overline{x}_B$ would belong to the same indifference class $I':=\left\{x \in X\middle| x \sim \overline{x}_B\right\}$. Since $I' \subseteq D$, the sets $I' \cap A$ (containing $\underbar{x}_C$) and $I' \cap B$ (containing $\overline{x}_B$) would constitute a disconnection of $I'$, a contradiction. Thus, $\underbar{x}_C \succ \overline{x}_B$.

Since $\succsim$ has no jumps, there exists $z \in X$ such that $\underbar{x}_C \succ z \succ \overline{x}_B$. It is easy to verify that $z \not \in A$ (otherwise $z \in C$, contradicting  $\underbar{x}_C \succ z$) and $z \not \in B$ (since $z \succ \overline{x}_B$ and $\overline{x}_B$ is the best element of $B$). Hence, $z \not \in D$. 

We will now show that $z$ is comparable to every alternative in $D$. Pick any $x \in D$. There are three possibilities: $x \in B$, $x \in C$, and $x \in A \setminus C$. If $x \in B$, since $z \succ \overline{x}_B$, we have $z \succ x$. If $x \in C$, since $\underbar{x}_C \succ z$, we also have $x \succ z$. Finally, if $x \in A \setminus C$, then the definition of $C$ and the assumption that $D$ is a domain imply that $\overline{x}_B \succ x$. It follows that $z \succ x$. This establishes that $z$ is comparable to every $x \in B \cup C \cup (A\setminus C) = D$, as claimed.

Since $z \not\in D$ and $z$ is comparable to every alternative in $D$, we conclude that $D \cup \{z\}$ is also a domain, contradicting the assumption that $D$ is a maximal domain.
\end{proof}

\begin{proof}[Proof of Lemma~\ref{lemma:betweenness}]
Take any $x,y \in X$ such that $x \sim y$ and $\alpha \in [0,1]$. Define $z := \alpha x+ (1-\alpha)y$. Since $x \succsim y$ and  $y \succsim x$, by betweenness, we have $x \succsim z \succsim y$ and $y \succsim z \succsim x$ and, so $z \sim y$. It follows that each indifference class is convex, thus connected.

It remains to show that $\succsim$ has no jumps. Suppose, seeking a contradiction, that there exist $x, y \in X$ such that $x \succ y$ but there is no $z$ with $x \succ z \succ y$. Consider the line segment $L = \{\lambda x + (1-\lambda)y \mid \lambda \in [0,1]\}$. Since $X$ is convex, $L \subseteq X$. By betweenness, for all $z \in L$, $x \succsim z \succsim y$. Since there are no intermediate strict preferences, for every $z \in L$, either $z \sim x$ or $z \sim y$.
Define $L_x = \{z \in L \mid z \sim x\}$ and $L_y = \{z \in L \mid z \sim y\}$. Since $\succsim$ is continuous, $L_x$ and $L_y$ are closed in $L$. They are disjoint (since $x \succ y$), nonempty (containing $x$ and $y$), and $L_x \cup L_y = L$. This implies $L$ is disconnected, a contradiction.
\end{proof}

\begin{proof}[Proof of Theorem~\ref{th:characterization}]
We start by establishing sufficiency through the following lemma:

\begin{lemma}
\label{lemma:domain_maxdomain}
Every connected domain that contains every indifferent alternative and has no exterior bound is a maximal domain.
\end{lemma}

\begin{proof}[Proof of Lemma~\ref{lemma:domain_maxdomain}]
Suppose, seeking a contradiction, that $D$ is a connected domain that contains every indifferent alternative, has no exterior bound, but is not a maximal domain. Then, by Lemma 1 in \citet{Gorno2018}, there exists a maximal domain $D'$ such that $D \subset D'$. Take $x \in D' \setminus D$. Since $D$ has no exterior bounds there are $y, z \in D$ such that $y \succ x \succ z$. Define $D_1 := \left\{w \in D \middle| w \succsim x\right\}$ and $D_2 := \left\{w \in D \middle| x \succsim w\right\}$. $D_1$ and $D_2$ are nonempty since $y \in D_1$ and $z \in D_2$. Also, $D_1 \cup D_2 = D$ because $x \in D'$ and $D'$ is a domain that contains $D$. Moreover, $D_1 \cap D_2 = \emptyset$. If this intersection were not empty, there would be $w \in D$ such that $x \sim w$, which would contradict that $D$ contains every indifferent alternative. Finally, $D_1$ and $D_2$ are closed relative to $D$. We conclude that $D$ is not connected, which is a contradiction.
\end{proof}

Now we turn to necessity. It is easy to show that every maximal domain contains every indifferent alternative and has no exterior bound. Moreover, since $\succsim$ satisfies the assumptions of Theorem~\ref{th:sufficiency}, every maximal domain is connected. 
\end{proof}

\begin{proof}[Proof of Theorem~\ref{th:sufficiency_arc}]
Let $D$ be a maximal domain. Since $X$ is second-countable, $D$ is second-countable (as a subspace). Since $\succsim$ is continuous and $X$ is compact, Theorem 1 in \citet*{Gorno2018} implies that $D$ is compact. Because $\succsim$ is complete and continuous on $D$, there exists a continuous utility representation $u:D \to \mathbb{R}$. 

Since $\succsim$ is antisymmetric, its indifference classes are singletons, hence connected. By Theorem~\ref{th:sufficiency}, $D$ is connected. It follows that $u(D)$ is connected and compact, thus a compact interval. Without loss of generality, we can assume that $u(D)=[0,1]$. Since $\succsim$ is antisymmetric, $u$ is a continuous bijection. Since $X$ is compact and $[0,1]$ is Hausdorff, $u$ is actually a homeomorphism between $D$ and $[0,1]$. It follows that $D$ is arc-connected, as desired.
\end{proof}

\begin{proof}[Proof of Proposition~\ref{prop:arc-connected}]
$X$ is a compact metrizable space (it is metrized by the L\'{e}vy metric) and $\geq_1$ is continuous, antisymmetric, and has no jumps. Thus, by Theorem~\ref{th:sufficiency_arc}, $\geq_1$ is arc-connected. In fact, the argument in the proof of Theorem~\ref{th:sufficiency_arc} shows that, in this setting, a subset of $X$ is a connected domain of $\geq_1$ that has no exterior bound if and only if it is the image of a $\geq_1$-increasing arc joining the degenerate CDFs associated with $0$ and $\overline{z}$. Hence, Theorem~\ref{th:characterization} implies the desired equivalence. 
\end{proof}

\begin{proof}[Proof of Proposition~\ref{prop:Schmeidler}]
To prove necessity, assume that $\succsim$ has no jumps. Since $\succsim$ is a Schmeidler preference, Proposition 10 in \citet*{Gorno2018} implies that every nontrivial connected component is contained in a maximal domain. Moreover, because $\succsim$ is a continuous preference on a compact space with connected indifference classes and no jumps, every maximal domain is connected by Theorem~\ref{th:characterization}. It follows that every nontrivial maximal domain is a connected component. Finally, since trivial maximal domains are indifference classes by definition, $\succsim$ is decomposable.

For sufficiency, note that, since $\succsim$ is decomposable and has connected indifference classes, every maximal domain is connected. Thus, Theorem~\ref{th:necessity} implies that $\succsim$ has no jumps.
\end{proof}

\section{On the compactness assumption}
\label{app:compactness}

Most of our results assume that the space $X$ is compact. This appendix provides a justification for this assumption by establishing an ``inverse'' result: if $X$ is a $T_1$ space where compactness is equivalent to countable compactness, and $X$ fails to be compact, we can always construct a preference that satisfies the conditions of Theorems~\ref{th:sufficiency} and \ref{th:sufficiency_arc} but fails to be connected. Because virtually all spaces of interest in economics---including all metrizable, all compact, and all second-countable spaces---satisfy this topological equivalence,\footnote{A notable exception---often arising in the analysis of infinite-horizon economies---is the dual of $L^\infty(\Omega, \Sigma, \mu)$, endowed with the weak* topology. If $L^\infty(\Omega, \Sigma, \mu)$ is infinite-dimensional, its dual---the space of bounded finitely additive measures---is not separable and admits subsets that are weak*-countably compact but strictly not weak*-compact.} the result below implies that compactness cannot be fruitfully relaxed for the general connectedness results we derive in this paper.

\begin{proposition}
\label{prop:non-compact}
Let $X$ be a $T_1$ space in which compactness is equivalent to countable compactness. If $X$ is not compact, then there exists a continuous antisymmetric preference on $X$ with no jumps that is not connected.
\end{proposition}

\begin{proof}
Because compactness is equivalent to countable compactness on $X$, the fact that $X$ is not compact implies that it is not countably compact. In a $T_1$ space, a space fails to be countably compact if and only if it contains an infinite subset $Y$ that has no limit points (accumulation points) in $X$. Because $Y$ has no limit points, it must be closed, and its relative topology must be discrete. Moreover, since $Y$ is infinite, we can extract a countably infinite subset; without loss of generality, assume $Y$ itself is countably infinite.

Let $\mathbb{D} = \{ m/2^k \mid m, k \in \mathbb{Z}, k \geq 0 \} \cap (0, 1)$ be the set of dyadic rationals in $(0, 1)$. Since $Y$ is a countably infinite set, there exists a bijection $f: Y \to \mathbb{D}$. We define the binary relation $\succsim$ on $X$ as follows:
\begin{enumerate}
    \item For any $x, y \in Y$, $x \succsim y$ if and only if $f(x) \geq f(y)$.
    \item For any $x, y \in X \setminus Y$, $x \succsim y$ if and only if $x = y$.
    \item For any $x \in Y$ and $y \in X \setminus Y$, $x$ and $y$ are incomparable.
\end{enumerate}

We first show that $\succsim$ is continuous. The preference is continuous if for every $x \in X$, the upper contour set $U(x) := \{y \in X \mid y \succsim x\}$ and the lower contour set $L(x) := \{y \in X \mid x \succsim y\}$ are closed in $X$.
If $x \in X \setminus Y$, then $U(x) = \{x\}$ and $L(x) = \{x\}$. Since $X$ is $T_1$, singleton sets are closed.
If $x \in Y$, then $U(x) = \{y \in Y \mid f(y) \geq f(x)\}$ and $L(x) = \{y \in Y \mid f(x) \geq f(y)\}$. Both sets are subsets of $Y$. Since $Y$ is closed and discrete, any subset of $Y$ is closed in $X$. In particular, $U(x)$ and $L(x)$ are closed subsets of $X$. Hence, $\succsim$ is continuous.

Next, we show that $\succsim$ is an antisymmetric preference. Reflexivity and transitivity are immediate. For antisymmetry, suppose $x \succsim y$ and $y \succsim x$. If $x, y \in Y$, then $f(x) \geq f(y)$ and $f(y) \geq f(x)$ imply $f(x) = f(y)$. Since $f$ is a bijection, $x = y$. If $x, y \in X \setminus Y$, the definition implies $x = y$. Thus, $\succsim$ is antisymmetric.

We now show that $\succsim$ has no jumps. A preference has a jump if there exist $x, y$ such that $x \succ y$ and there is no $z$ with $x \succ z \succ y$. If $x \succ y$, then by definition $x, y \in Y$ and $f(x) > f(y)$. Since $\mathbb{D}$ is dense in $(0,1)$, there exists $q \in \mathbb{D}$ such that $f(x) > q > f(y)$. Since $f$ is surjective onto $\mathbb{D}$, there exists $z \in Y$ such that $f(z) = q$. Therefore, $x \succ z \succ y$, so $\succsim$ has no jumps.

Finally, we show that the preference $\succsim$ is not connected. $Y$ is a maximal domain and has infinitely many elements. Since $Y$ has no limit points, it carries the discrete topology as a subspace of $X$. A discrete space with more than one element is not connected. Thus, the preference is not connected.
\end{proof}

The topological assumptions on $X$ in Proposition~\ref{prop:non-compact} are tight. If either the $T_1$ separation axiom or the equivalence between compactness and countable compactness is dropped, the converse statement fails. One can construct pathological non-compact spaces where either the topology prevents the continuous embedding of linear orders (making the construction of a disconnected preference impossible), or where all subsets are inherently connected (making all preferences trivially connected).\footnote{The broadest standard class of topological spaces where compactness and countable compactness coincide is the class of meta-Lindel\"of spaces \citep{ArensDugundji1950}. If this equivalence is dropped, we can consider the space of countable ordinals $\omega_1$, which is a $T_1$ space that is countably compact but not compact. Because $\omega_1$ is scattered (every nonempty subset contains an isolated point), it cannot continuously embed a dense linear order, implying that any nontrivial continuous preference must exhibit jumps. Alternatively, if the $T_1$ axiom is dropped, the set of natural numbers endowed with the nested interval topology satisfies the compactness equivalence (as it is second-countable) but has all its subsets connected, rendering the maximal domains of any preference trivially connected.} 

However, because the class of $T_1$ spaces where compactness coincides with countable compactness subsumes the vast majority of choice spaces utilized in economic applications---including all metrizable, second-countable, Lindel\"of, and paracompact spaces \citep[see][Theorems 2.50, 2.51, and 2.52]{AliprantisBorder2006}---Proposition~\ref{prop:non-compact} clearly delineates the boundaries of Theorems~\ref{th:sufficiency} and \ref{th:sufficiency_arc}. The message is that, for most economic models, the assumption that $X$ is compact is tight: without it, the remaining hypotheses fail to guarantee that the preference structure is connected (or arc-connected).

\bibliographystyle{abbrvnat}
\bibliography{refs}

@book{AliprantisBorder2006,
  title={Infinite Dimensional Analysis: A Hitchhiker's Guide},
  author={Aliprantis, Charalambos D and Border, Kim C},
  year={2006},
  edition={3rd},
  publisher={Springer},
  address={Berlin}
}

@article{ArensDugundji1950,
  title={Remark on the concept of compactness},
  author={Arens, Richard and Dugundji, James},
  journal={Portugaliae Mathematica},
  volume={9},
  number={4},
  pages={141--143},
  year={1950}
}

@article{Dekel1986,
author = {Dekel, Eddie},
title = {An axiomatic characterization of preferences under uncertainty: Weakening the independence axiom},
journal = {Journal of Economic Theory},
volume = {40},
number = {2},
pages = {304 - 318},
year = {1986}
}

@ARTICLE{Dubra2004,
author = {Dubra, Juan and Maccheroni, Fabio and Ok, Efe A.},
title = {Expected Utility Theory without the Completeness Axiom},
journal = {Journal of Economic Theory},
year = {2004},
volume = {115},
number = {1},
pages = {118-133}
}

@ARTICLE{Eilenberg1941,
author={Eilenberg, Samuel},
title={Ordered topological spaces},
journal={American Journal of Mathematics},
year={1941},
volume={63},
number={1},
pages={39-45}
}

@ARTICLE{Gorno2018,
author={Gorno, Leandro},
title={The structure of incomplete preferences},
journal={Economic Theory},
year={2018},
volume={66},
number={1},
pages={159-185}
}

@ARTICLE{GornoRivello2023,
author={Gorno, Leandro and Rivello, Alessandro T.},
title={A maximum theorem for incomplete preferences},
journal={Journal of Mathematical Economics},
year={2023},
volume={106},
pages={102822}
}

@ARTICLE{Gul1991,
author={Gul, Faruk},
title={A theory of disappointment Aversion},
journal={Econometrica},
year={1991},
volume={59},
number={3},
pages={667-686}
}

@ARTICLE{KhanUyanik2021,
author={Khan, M. Ali and Uyanık, Metin},
title={Topological connectedness and behavioral assumptions on preferences: a two-way relationship},
journal={Economic Theory},
volume={71},
number={2},
pages={411-460},
year={2021}
}

@BOOK{Ok2007,
author={Ok, Efe A.},
title={Real analysis with economic applications},
publisher={Princeton University Press},
year={2007}
}

@BOOK{Rader1972,
  author={Rader, Trout},
  title={Theory of Microeconomics},
  year={1972},
  publisher={Academic Press}
}

@ARTICLE{Schmeidler1971,
author={Schmeidler, David},
title={A condition for the completeness of partial preference relations},
journal={Econometrica},
year={1971},
volume={39},
number={2},
pages={403-404}
}

@ARTICLE{Sen1969,
author={Sen, Amartya},
title={Quasi-transitivity, rational choice and collective decisions},
journal={The Review of Economic Studies},
year={1969},
volume={36},
number={3},
pages={381-393}
}

@ARTICLE{Sonnenschein1965,
author={Sonnenschein, Hugo},
title={The relationship between transitive preference and the structure of the choice space},
journal={Econometrica},
year={1965},
pages={624-634}
}

\end{document}